\documentclass[11pt]{article}

\usepackage[utf8]{inputenc}
\usepackage[T1]{fontenc}
\usepackage[UKenglish]{babel}
\usepackage{lmodern}
\usepackage{microtype}
\usepackage{amsmath,amssymb,amsthm}
\usepackage{graphicx}
\usepackage{xcolor}
\usepackage{url}
\usepackage{comment}
\usepackage{hyperref}
\usepackage[nameinlink,noabbrev]{cleveref}

\newcommand{\CJreconf}{\ensuremath{\overset{CJ}\leftrightarrow}}
\newcommand{\CSreconf}{\ensuremath{\overset{CS}\leftrightarrow}}
\newcommand{\CSSreconf}{\ensuremath{\overset{CS1}\leftrightarrow}}
\newcommand{\CCRCJ}{\ensuremath{\mathrm{CCR\_ CJ}}}
\newcommand{\CCRCS}{\ensuremath{\mathrm{CCR\_ CS}}}

\newcommand{\sizeMS}{\ensuremath{\text{size}(MS)}}
\newcommand{\sizeML}{\ensuremath{\text{size}(ML)}}

\theoremstyle{plain}
\newtheorem{theorem}{Theorem}
\newtheorem{lemma}{Lemma}

\theoremstyle{definition}
\newtheorem{definition}{Definition}

\newcommand{\keywords}[1]{\def\paperkeywords{#1}}
\newcommand{\printkeywords}{%
  \begin{center}
    \small\textbf{Keywords:} \paperkeywords
  \end{center}
}

\title{NP-Hardness of Connected Components Reconfiguration under Component Jumping on Caterpillar Graphs}

\author{%
Naoki Kitamura\\
\small The University of Osaka, Japan\\
\small\texttt{n-kitamura@ist.osaka-u.ac.jp}
\and
Seitaro Kawaguchi\\
\small The University of Osaka, Japan
\and
Yuya Terashima\\
\small The University of Osaka, Japan
\and
Taisuke Izumi\\
\small The University of Osaka, Japan
}
\date{}


\begin{document}

\maketitle

\begin{abstract}
We study the Connected Components Reconfiguration problem (CCR), in which connected components on a graph are transformed according to a specified reconfiguration rule. CCR generalizes Independent Set Reconfiguration by treating tokens not as individual vertices but as connected components of prescribed sizes. Among the variants of CCR, we focus on the component-jumping model, denoted by \CCRCJ. Nakahata.\ introduced this problem and showed that the decision problem for \CCRCJ~can be solved in $O(n^2)$ time on path graphs for arbitrary component sizes, and in polynomial time on chordal graphs when all connected components have the same size. However, the complexity on chordal graphs under a multiset size constraint remained open.

In this paper, we study this multiset version of \CCRCJ~from both complexity-theoretic and algorithmic viewpoints. First, we prove that \CCRCJ~is NP-hard even on caterpillar graphs, which is a very restricted subclass of trees and chordal graphs minimally above path graphs.  This result immediately implies NP-hardness for chordal graphs under a multiset size constraint, thereby resolving Nakahata's open problem on chordal graphs under multiset size constraints. Second, we  revisit \CCRCJ~on path graphs. We improve the previous $O(n^2)$-time algorithm for the decision problem by giving an $O(n\log n)$-time decision algorithm. Moreover, when the instance has sufficiently large empty space, we show that there exists a reconfiguration sequence of length $O(n\log n)$, and such a sequence can be output efficiently.
\end{abstract}

\printkeywords

\section{Introduction}
Combinatorial reconfiguration problems have attracted considerable attention in discrete algorithms and computational complexity. A reconfiguration problem asks whether there exists a step-by-step transformation between two feasible solutions according to a specified reconfiguration rule. Such problems have been studied for various combinatorial objects, including independent sets, cliques, vertex covers, vertex colorings, and matchings~\cite{BonamyBHIKMMW19,ItoDHPSUU11,ItoNZ16, ItoOO23, Nishimura18,Heuvel13}.

Among reconfiguration problems, those for independent sets are among the most extensively studied~\cite{BousquetMNS22,KaminskiMM12}. Given a graph $G=(V,E)$, an independent set is a vertex subset whose vertices are pairwise nonadjacent. In the reconfiguration setting, each vertex in an independent set is regarded as containing a token, and one reconfiguration step moves a single token so that the resulting vertex set remains an independent set. Two of the best-studied rules are token sliding (TS), where a token may move only to an adjacent vertex, and token jumping (TJ), where a token may move to any vertex in the graph~\cite{KaminskiMM12}. The Independent Set Reconfiguration problem under both TS and TJ is known to be PSPACE-complete~\cite{KaminskiMM12,LokshtanovM19,Wrochna20}. This has motivated extensive research on restricted graph classes~\cite{BelmonteKLMOS21,BousquetMNS22,KaminskiMM12,LokshtanovM19,Wrochna20}, as well as on other variants, such as bounded-hop models~\cite{hatano2026boundedhop} and models allowing simultaneous movement of multiple tokens~\cite{hirahara2025reachability,kvrivstan2025reconfiguration}.

Recently, the Connected Components Reconfiguration problem (CCR) was introduced as a generalization of Independent Set Reconfiguration~\cite{nakahata2025reconfiguring}. In this problem, instead of independent sets, one is given a collection of connected components on a graph. Initially, these connected components are pairwise nonadjacent and vertex-disjoint. One reconfiguration step consists of moving one connected component to another connected vertex set, under the restriction that the moved component must remain vertex-disjoint from and nonadjacent to the other connected components. As in Independent Set Reconfiguration, the following two reconfiguration rules were proposed.
\begin{itemize}
    \item Component Sliding (CS): A connected component can be moved only when the union of the connected components before and after the move is connected.
    \item Component Jumping (CJ): A connected component can be moved to any connected vertex set in the graph.
\end{itemize}

When every connected component has size $1$, the component-sliding model coincides with token sliding for Independent Set Reconfiguration, and the component-jumping model coincides with token jumping. Nakahata showed that the component-sliding model can be solved in polynomial time on path graphs and cographs~\cite{nakahata2025reconfiguring}. For the component-jumping model, they showed that the decision problem for \CCRCJ~can be solved in $O(n^2)$ time on path graphs for arbitrary component sizes, and in polynomial time on chordal graphs when all connected components have the same size. However, the complexity on chordal graphs under a multiset size constraint remained open.

Very recently, Otachi and Toyoda~\cite{otachi2026biclique} studied a related hardness question via biclique reconfiguration. They proved that Balanced Biclique Reconfiguration is PSPACE-complete on bipartite graphs, and used this result to show that Connected Components Reconfiguration with two connected components of the same size is PSPACE-complete under both CJ and CS, even on co-bipartite graphs. This gives a negative answer to a different open problem of Nakahata, concerning whether CCR becomes tractable when the number of connected components is bounded by a constant.

In this paper, we study this multiset version of \CCRCJ~from both complexity-theoretic and algorithmic viewpoints. Our first result concerns caterpillar graphs, which are a simple subclass of trees and hence also of chordal graphs.

\begin{theorem}
\label{thm:main}
    The Connected Components Reconfiguration problem in the component-jumping model is NP-hard for caterpillar graphs.
\end{theorem}

Since caterpillar graphs are a subclass of chordal graphs, this immediately implies NP-hardness for chordal graphs under a multiset size constraint, thereby resolving Nakahata's open problem on chordal graphs under multiset size constraints.

Our second result improves the algorithmic complexity of the decision problem for \CCRCJ~on path graphs.

\begin{theorem}
\label{thm:intro-path}
    The decision problem for \CCRCJ~on path graphs can be solved in $O(n\log n)$ time.
\end{theorem}

We also show that, when the instance has sufficiently large empty space, there exists a reconfiguration sequence of length $O(n\log n)$, and such a sequence can be output algorithmically. This result is stated formally as Theorem~\ref{thm:myconstruct}.

\section{Preliminary}
\subsection{Graph}
In this paper, we consider a simple undirected unweighted graph $G=(V(G),E(G))$. Let $n=|V(G)|$ and $m=|E(G)|$. When the graph $G$ is clear from the context, we simply write $V$ and $E$ instead of $V(G)$ and $E(G)$, respectively. 
We define $m(U)$ as the multiset of the sizes of connected components in the induced subgraph $G[U]$, and call it the \emph{connected-component multiset} of $U$ in $G$.

A caterpillar graph is a tree obtained from a path by attaching degree-one vertices. More precisely, a caterpillar graph can be decomposed into a vertex set $S$ (the spine) and a vertex set $L$ (the legs), where $G[S]$ is a path graph, every vertex in $L$ has degree one, and every vertex in $L$ is adjacent to a vertex in $S$. For a vertex $v \in S$, we define a leg adjacent to $v$ to be a \emph{leaf attached to $v$}. Note that a vertex in the spine may itself have degree one, but such a vertex is not regarded as a leaf attached to the spine. An example of a caterpillar graph is shown in Figure~\ref{fig:0}.

\begin{figure}[t]
  \centering
  \includegraphics[width=0.5\textwidth]{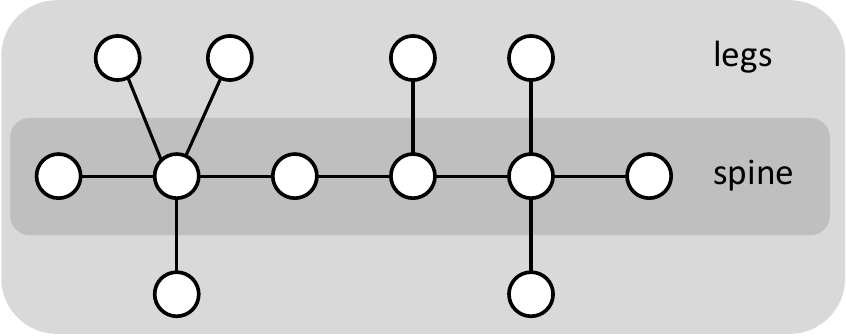}
  \caption{An example of caterpillar graph.}
  \label{fig:0}
\end{figure}

\subsection{Connected Components Reconfiguration (CCR)}
The connected components reconfiguration problem is defined as follows.

\begin{definition}[Connected Components Reconfiguration (CCR)]
Given as input a graph $G$, vertex subsets $A,B\subseteq V$, a multiset of positive integers $\mathcal{M}$, and a reconfiguration rule $R$, the connected components reconfiguration problem asks whether there exists a sequence of vertex subsets from $A$ to $B$ such that
\begin{enumerate}
    \item the connected-component multiset of every vertex subset in the sequence is equal to $\mathcal{M}$, and
    \item every two consecutive vertex subsets in the sequence are adjacent under $R$.
\end{enumerate}
\end{definition}

As for the reconfiguration rule $R$, we consider component jumping (CJ) and component sliding (CS). A reconfiguration rule defines the adjacency relation between vertex subsets. For $U,U'\subseteq V(G)$ with $m(U)=m(U')$, let $C(U)$ and $C(U')$ denote the sets of connected components of $U$ and $U'$, respectively. We write $U \CJreconf U'$ and $U \CSreconf U'$ to denote that $U$ and $U'$ are adjacent under CJ and CS, respectively. Furthermore, when $|C(U)\setminus C(U')| = |C(U')\setminus C(U)|=1$, let $C$ and $C'$ denote the unique connected components in $C(U)\setminus C(U')$ and $C(U')\setminus C(U)$, respectively. Then the relations \CJreconf, \CSreconf, and \CSSreconf~are defined as follows:
\begin{align*}
    & U \CJreconf U' \Leftrightarrow |C(U)\setminus C(U')| = |C(U')\setminus C(U)| = 1, \\
    & U \CSreconf U' \Leftrightarrow U \CJreconf U' \wedge C \cup C' \text{ is connected}, \\
    & U \CSSreconf U' \Leftrightarrow U \CSreconf U' \wedge |C \setminus C'| = |C' \setminus C| = 1.
\end{align*}

The connected components reconfiguration problem under the rule \CSreconf~is called \CCRCS, and that under the rule \CJreconf~is called \CCRCJ. In this paper, we consider only \CCRCJ. Observe that if the connected-component multisets of $A$ and $B$ are not equal to $\mathcal{M}$, then $A$ and $B$ are clearly not reconfigurable under \CCRCJ. Therefore, in the following, we represent an instance of \CCRCJ~simply by a triple $(G,A,B)$.

By the definition of \CCRCJ, the connected-component multiset must remain unchanged before and after each move. Equivalently, moving a connected component must satisfy the following conditions:
\begin{enumerate}
    \item The moved component induces a connected subgraph of the same size as the original component.
    \item The moved component shares no vertex with any other connected component.
    \item The moved component is not adjacent to any other connected component.
\end{enumerate}

\subsection{Independent Set Problem}
For a vertex subset $U\subseteq V$, if $G[U]$ has no edges, then $U$ is called an independent set of $G$. Given as input a graph $G$ and a positive integer $k$, the problem of deciding whether $G$ has an independent set of size $k$ is called the Independent Set problem. A graph in which every vertex has degree $h$ is called an $h$-regular graph. The following fact is known for $3$-regular graphs.
\begin{theorem}[\cite{fleischner2010maximum}]
The Independent Set problem is NP-complete for $3$-regular graphs.
\end{theorem}

\section{Proof of the NP-Hardness of \CCRCJ~on Caterpillar Graphs}
In this section, we prove that \CCRCJ~on caterpillar graphs is NP-hard. The proof is given by a reduction from the Independent Set problem on $3$-regular graphs. More specifically, we reduce an instance $\Phi=(G,k)$ of the Independent Set problem on a $3$-regular graph to an instance $\Phi'=(G',A,B)$ of \CCRCJ~such that $\Phi'$ is reconfigurable if and only if $G$ has an independent set of size $k$.

\subsection{Construction of the Gadgets}
Let the vertex set of the instance $\Phi$ be $V=\{v_1,v_2,\dots,v_{n}\}$, and let the edge set be $E=\{e_1,e_2,\dots,e_{m}\}$. Let $\pi^{-1}:E\rightarrow \mathbb{N}$ be the function that maps each edge $e_i$ to $i$. For each vertex $v$, we denote the three incident edges by $e_{v,1},e_{v,2},e_{v,3}$. Although the graphs considered in this paper are undirected, we assume that each edge is assigned an orientation for the purpose of the reduction. For an edge $e$ and a vertex $v$ incident to $e$, we define $\rho_e(v)$ as follows: if $v$ is the tail of $e$, then $\rho_e(v)=0$, and if $v$ is the head of $e$, then $\rho_e(v)=1$. To carry out the reduction, we construct the following gadgets.
\begin{itemize}
    \item $n$ independent-set gadgets $\{I_1,I_2,\dots,I_n\}$.
    \item $n$ vertex-cover gadgets $\{VC_1,VC_2,\dots,VC_n\}$.
    \item $2m$ edge gadgets $\{Ed_{1,a},Ed_{2,a},\dots,Ed_{m,a},Ed_{1,b},Ed_{2,b},\dots,Ed_{m,b}\}$.
    \item $k$ independent-set padding gadgets $\{MS_1,MS_2,\dots,MS_k\}$.
    \item $n-k$ vertex-cover padding gadgets $\{ML_1,ML_2,\dots,ML_{n-k}\}$.
    \item one primary gadget $L_1$.
\end{itemize}

In addition, between gadgets we prepare separators, which are used to separate the gadgets from one another. Below, we explain each gadget.

\paragraph*{Independent-Set Gadget}
The independent-set gadget $I_i$ is constructed as follows.
\begin{itemize}
    \item Prepare a path (spine) $P=\{p_1,p_2,\dots,p_{11}\}$ on $11$ vertices.
    \item For each of the vertices $p_2,p_4,p_6,p_8,p_{10}$, the number of attached leaves is set to $0$.
    \item For any $j\in\{1,2,3\}$, the number of leaves attached to the vertex $p_{4j-3}$ is set to $500\pi^{-1}(e_{v_i,j})m^{10}+m+\rho_{e_{v_i,j}}(v_i)$.
    \item For any $j\in\{1,2,3\}$, the number of leaves attached to the vertex $p_{4j-1}$ is set to $500(3m-\pi^{-1}(e_{v_i,j}))m^{10}+m+1-\rho_{e_{v_i,j}}(v_i)$.
\end{itemize}

An illustration of the independent-set gadget is shown in Figure~\ref{fig:1}. Hereafter, for a vertex $p_j$ contained in the independent-set gadget $I_i$, we write it as $p_{j,I_i}$.

  \begin{figure}[t]
  \centering
  \includegraphics[width=1.0\textwidth]{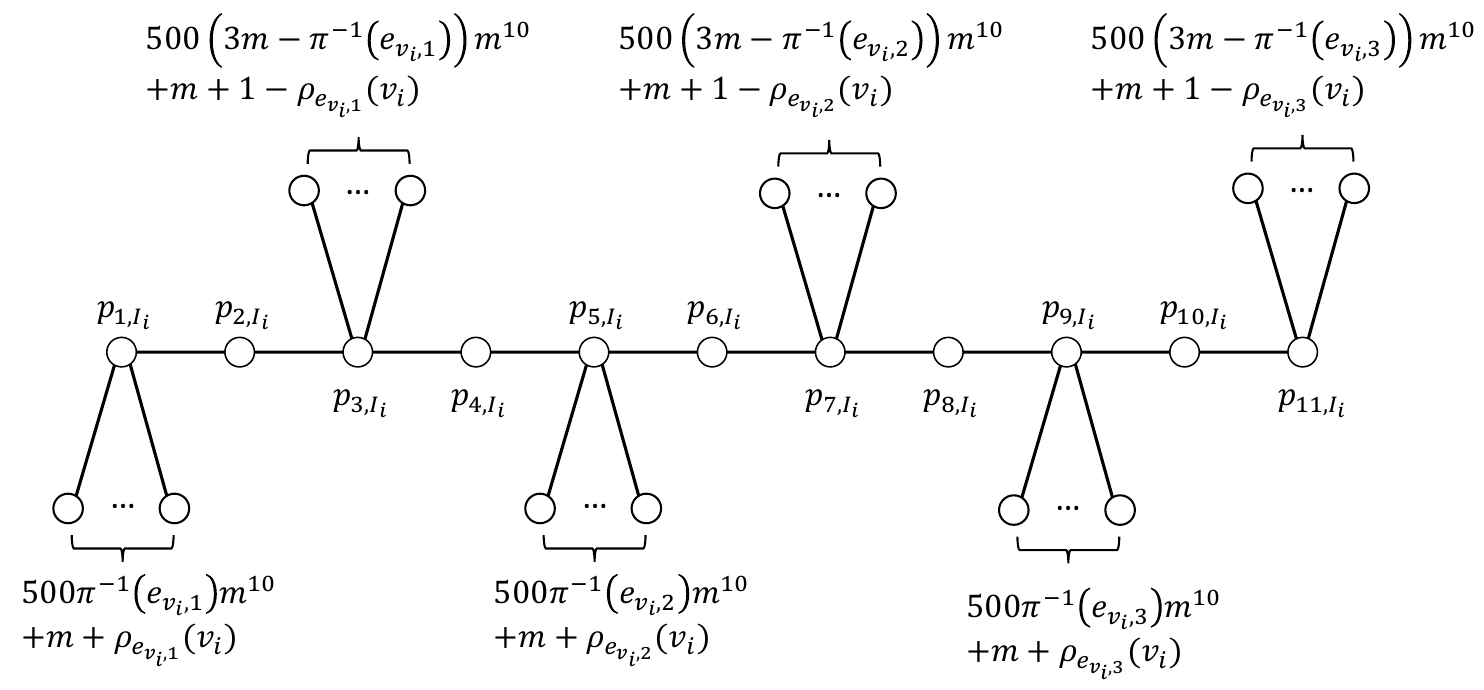}
  \caption{An example of the independent-set gadget $I_i$.}
  \label{fig:1}
  \end{figure}

\paragraph*{Vertex-Cover Gadget}
The vertex-cover gadget $VC_i$ is constructed as follows.
\begin{itemize}
    \item Prepare a path (spine) $P=\{p_1,p_2,\dots,p_{11}\}$ on $11$ vertices.
    \item For each of the vertices $p_2,p_4,p_6,p_8,p_{10}$, the number of attached leaves is set to $0$.
    \item For any $j\in\{1,2,3\}$, the number of leaves attached to the vertex $p_{4j-3}$ is set to $500\pi^{-1}(e_{v_i,j})m^{10}+\rho_{e_{v_i,j}}(v_i)$.
    \item For any $j\in\{1,2,3\}$, the number of leaves attached to the vertex $p_{4j-1}$ is set to $500(3m-\pi^{-1}(e_{v_i,j}))m^{10}+1-\rho_{e_{v_i,j}}(v_i)$.
\end{itemize}

An illustration of the vertex-cover gadget is shown in Figure~\ref{fig:2}. Hereafter, for a vertex $p_j$ contained in the vertex-cover gadget $VC_i$, we write it as $p_{j,VC_i}$.
\begin{figure}[t]
  \centering
  \includegraphics[width=1.0\textwidth]{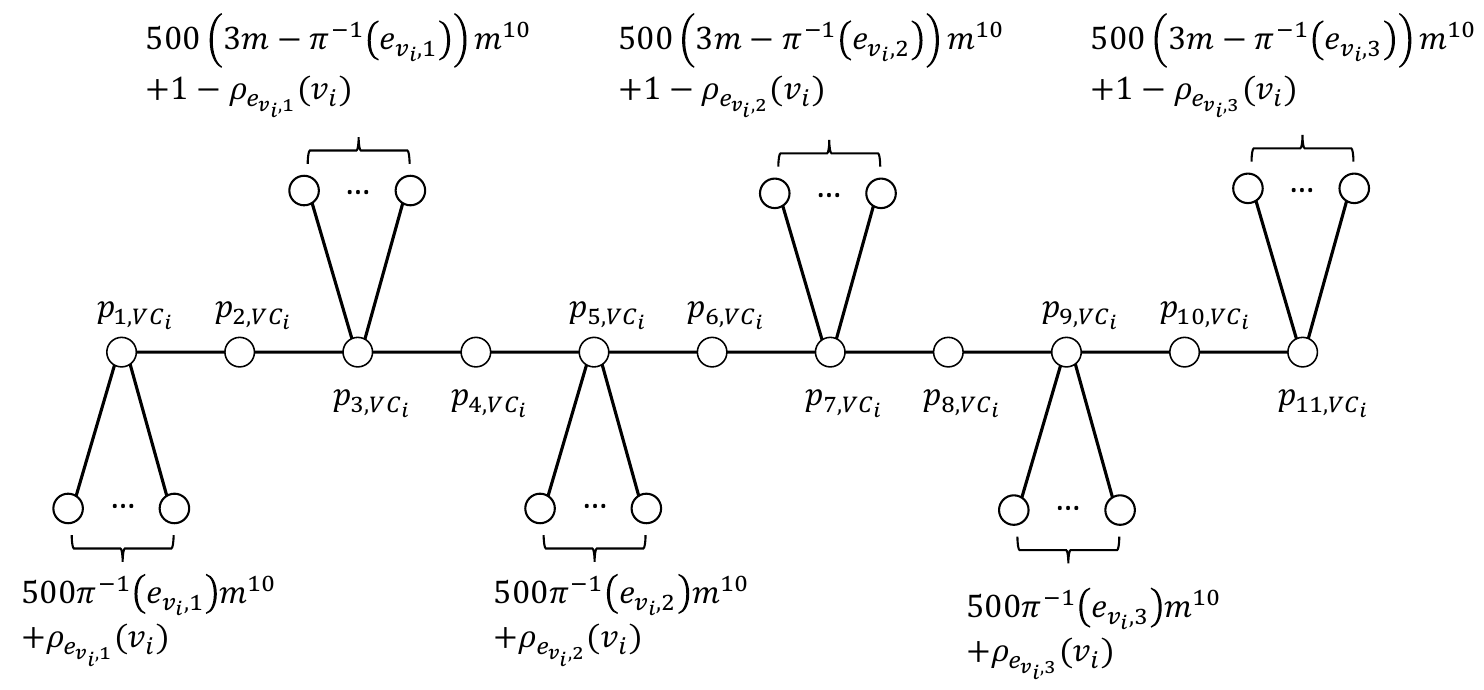}
  \caption{An example of the vertex-cover gadget $VC_i$.}
  \label{fig:2}
\end{figure}

\paragraph*{Edge Gadget}
The edge gadget $Ed_i$ is constructed as follows.
\begin{itemize}
    \item \sloppy{Prepare two paths (spines) $P_a=\{p_{1,a},p_{2,a},\dots,p_{5,a}\}$ and $P_b= \{p_{1,b} , p_{2,b} , \dots , p_{5,b}\}$ on $5$ vertices.}
    \item For each $x\in\{a,b\}$, the numbers of attached leaves of the vertices $p_{2,x}$ and $p_{4,x}$ are set to $0$. In addition, the number of attached leaves of the vertex $p_{3,x}$ is set to $m-2$.
    \item The numbers of attached leaves of the vertices $p_{1,a}$ and $p_{5,a}$ are set to $500im^{10}$ and $500im^{10}+1$, respectively.
    \item The numbers of attached leaves of the vertices $p_{1,b}$ and $p_{5,b}$ are set to $500(3m-i)m^{10}$ and $500(3m-i)m^{10}+1$, respectively.
\end{itemize}
We denote by $Ed_{i,a}$ the caterpillar graph consisting of the spine $P_a$ and its attached leaves, and by $Ed_{i,b}$ the caterpillar graph consisting of the spine $P_b$ and its attached leaves.

An illustration of the edge gadget is shown in Figure~\ref{fig:3}. Figure~\ref{fig:3}(a) represents $Ed_{i,a}$, and Figure~\ref{fig:3}(b) represents $Ed_{i,b}$. Hereafter, for a vertex $p_j$ contained in the edge gadget $Ed_{i,x}$, we write it as $p_{j,x,Ed_i}$, where $x\in\{a,b\}$.
\begin{figure}[t]
  \centering
  \includegraphics[width=1.0\textwidth]{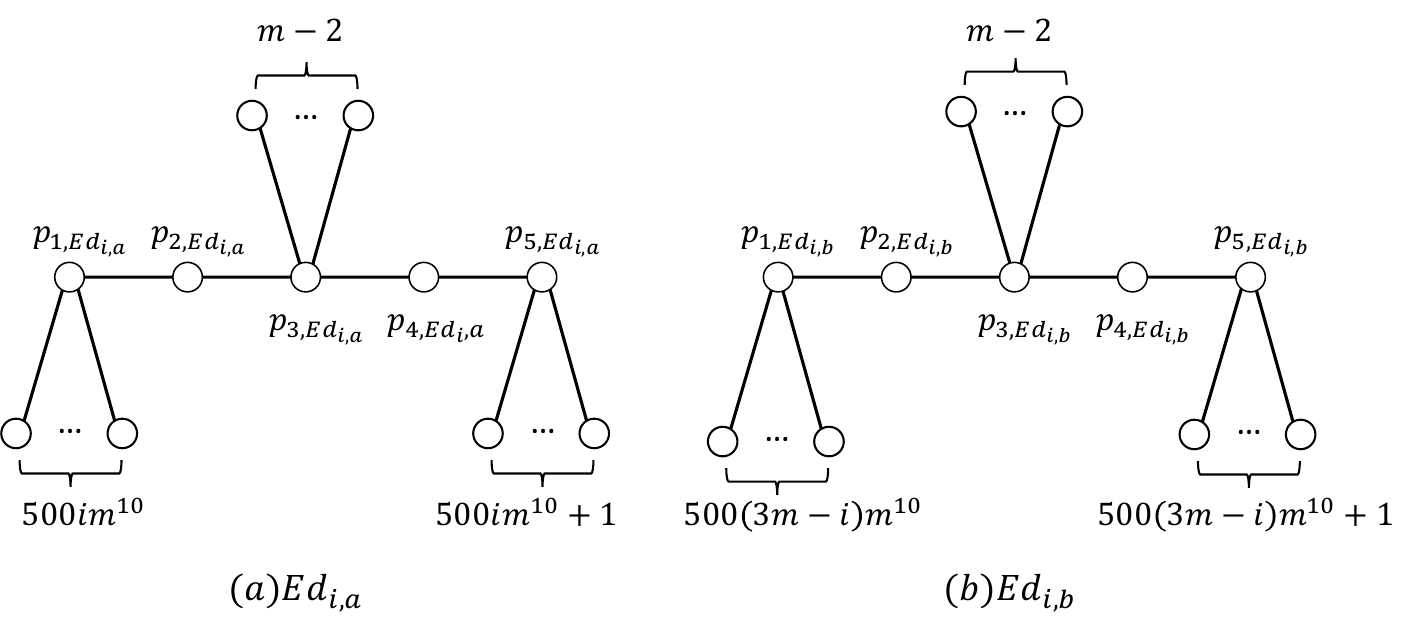}
  \caption{An example of the edge gadget $Ed_i$. (a) represents $Ed_{i,a}$, and (b) represents $Ed_{i,b}$.}
  \label{fig:3}
\end{figure}

\paragraph*{Independent-Set Padding Gadget}
The independent-set padding gadget $MS_i$ is constructed as follows.
\begin{itemize}
    \item Prepare a path (spine) $P=\{p_1,p_2\}$ on $2$ vertices.
    \item Let $\sizeMS=(500\cdot 3m\cdot m^{10}+2m+3)\cdot 3+6$.
    The number of attached leaves of the vertex $p_1$ is set to $\sizeMS-2=(500\cdot 3m\cdot m^{10}+2m+3)\cdot 3+4$.
    \item The number of attached leaves of the vertex $p_2$ is set to $0$.
\end{itemize}

An illustration of the independent-set padding gadget is shown in Figure~\ref{fig:4}.

\begin{figure}[t]
  \centering
  \includegraphics[width=0.7\textwidth]{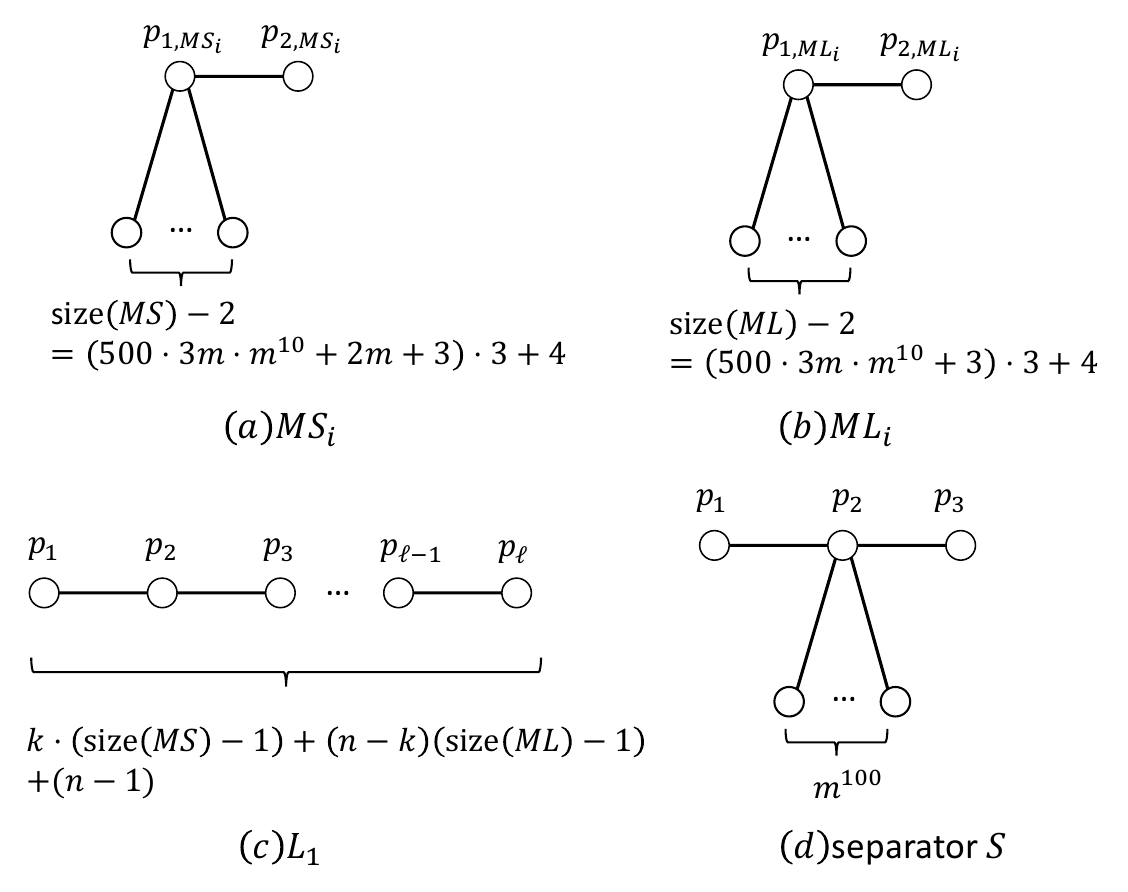}
  \caption{Examples of (a) the independent-set padding gadget $MS_i$, (b) the vertex-cover padding gadget $ML_i$, (c) the primary gadget, and (d) the separator. All independent-set padding gadgets are isomorphic, regardless of the value of $i$, and all vertex-cover padding gadgets $ML_i$ with $i \in \{1,\dots,n-k-1\}$ are isomorphic; note that $ML_{n-k}$ has a different form.
  }
  \label{fig:4}
\end{figure}

\paragraph*{Vertex-Cover Padding Gadget}
The vertex-cover padding gadget $ML_i$ is constructed as follows.
\begin{itemize}
    \item Prepare a path (spine) $P=\{p_1,p_2\}$ on $2$ vertices.
    \item Let $\sizeML=(500 \cdot 3m \cdot m^{10}+3)\cdot 3+6$.
    The number of attached leaves of the vertex $p_1$ is set to $\sizeML-2=(500 \cdot 3m \cdot m^{10}+3)\cdot 3+4$.
    \item The number of attached leaves of the vertex $p_2$ is set to $0$.
\end{itemize}

An illustration of the vertex-cover padding gadget is shown in Figure~\ref{fig:4}.
Hereafter, we collectively refer to the independent-set padding gadgets and the vertex-cover padding gadgets as padding gadgets.

\paragraph*{Primary Gadget}
The primary gadget $L_1$ is constructed as follows.
\begin{itemize}
    \item Let $
        \ell = k\cdot (\sizeMS-1)+(n-k)\cdot (\sizeML-1) +(n-1) = \left( ( 500 \cdot 3m \cdot m^{10} + 2m+3) \cdot  3 + 5 \right ) \cdot k + \left( (500\cdot 3m \cdot m^{10}+3) \cdot 3+5 \right) \cdot (n-k)+(n-1)$.
    Prepare a path (spine) $P=\{p_1,p_2,\dots,p_{\ell+1}\}$ of length $\ell$.
\end{itemize}

The length of the primary gadget is exactly equal to the total number of vertices in the independent-set padding gadgets and the vertex-cover padding gadgets minus $1$. An illustration of the primary gadget is shown in Figure~\ref{fig:4}.

\paragraph*{Separator}
The separator $S$ is constructed as follows.
\begin{itemize}
    \item Prepare a path (spine) $P=\{p_1,p_2,p_3\}$ on $3$ vertices.
    \item The number of attached leaves of the vertex $p_2$ is set to $m^{100}$.
\end{itemize}

An illustration of the separator is shown in Figure~\ref{fig:4}.

\subsection{Whole Construction}\label{sec:whole}
Each gadgets are connected to each other, or through separators, at their spine endpoints.
The overall construction is shown in Figure~\ref{fig:7}. For simplicity, the figure omits individual spine and leaf vertices and represents each gadget only by its name. The construction is described below.

\begin{figure}[t]
\centering
\includegraphics[width=1.0\textwidth]{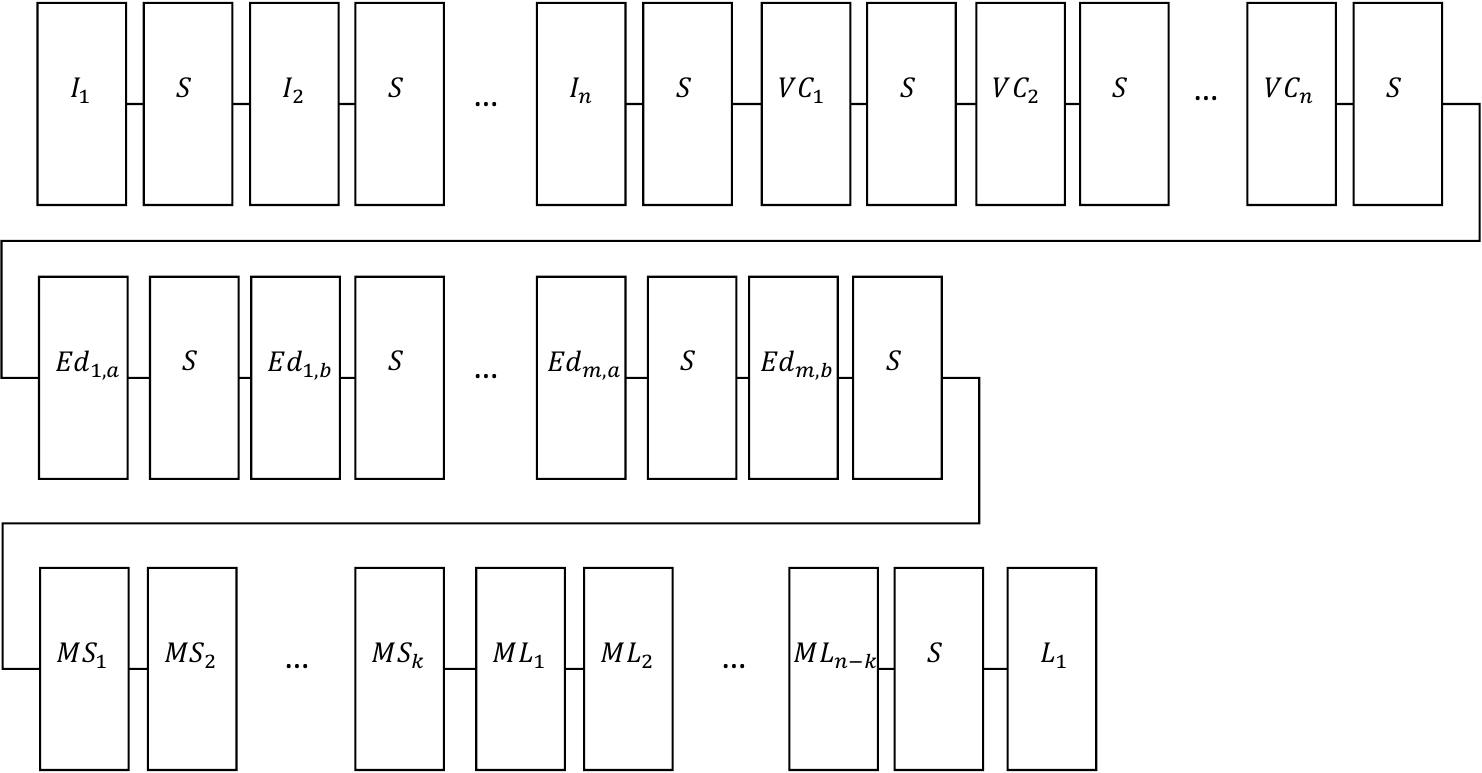}
\caption{The overall construction of the graph.}
\label{fig:7}
\end{figure}

\begin{itemize}
    \item Connect the $n$ independent-set gadgets in sequence. For each pair of consecutive gadgets, insert a separator gadget and connect an endpoint of the spine of one independent-set gadget to an endpoint of the spine of the separator gadget.
    \item Connect the $n$ vertex-cover gadgets in sequence. For each pair of consecutive gadgets, insert a separator gadget and connect an endpoint of the spine of one vertex-cover gadget to the spine of the separator gadget. Also insert a separator gadget between $I_n$ and $VC_1$, and connect the spine endpoints of these gadgets to the spine of the separator gadget.
    \item For each edge gadget $Ed_i$, insert a separator gadget between $Ed_{i,a}$ and $Ed_{i,b}$ and connect their spine endpoints ($1\leq i\leq m$). Then connect the $m$ edge gadgets in sequence. For each pair of consecutive edge gadgets, insert a separator gadget and connect an endpoint of the spine of one edge gadget to the spine of the separator gadget. Also insert a separator gadget between $VC_n$ and $Ed_{1,a}$, and connect the spine endpoints of these gadgets to the spine of the separator gadget.
    \item Connect the $k$ independent-set padding gadgets in sequence by connecting consecutive gadgets through their spine endpoints. In addition, insert a separator gadget between $Ed_m$ and $MS_1$, and connect the spine endpoints of these gadgets to the spine of the separator gadget.
    \item Connect the $n-k$ vertex-cover padding gadgets in sequence by connecting consecutive gadgets through their spine endpoints. Also connect the spine endpoint of $MS_k$ to the spine endpoint of $ML_1$.
    \item Insert a separator gadget between $ML_{n-k}$ and the primary gadget $L_1$, and connect the endpoint of each gadget to the spine of the separator gadget.
\end{itemize}

Next, we describe the positions of the connected components before reconfiguration.

\begin{itemize}
    \item For each independent-set gadget $I_i$ ($1\leq i\leq n$), place one connected component on each spine vertex in $\{p_1,p_3,p_5,p_7,p_9,p_{11}\}$ together with all leaves attached to that vertex.
    \item For each vertex-cover gadget $VC_i$ ($1\leq i\leq n$), place one connected component on each spine vertex in $\{p_1,p_3,p_5,p_7,p_9,p_{11}\}$ together with all leaves attached to that vertex.
    \item For each independent-set padding gadget $MS_i$ ($1\leq i\leq k$), place one connected component whose vertex set consists of the spine vertex $p_1$ and all leaves attached to $p_1$.
    \item For each vertex-cover padding gadget $ML_i$, place one connected component whose vertex set consists of the spine vertex $p_1$ and all leaves attached to $p_1$.
    \item \begin{sloppypar}Place one connected component in the primary gadget so that it occupies the entire spine of the gadget. Its size is
    \begin{align*}
        &\left( ( 500 \cdot 3m \cdot m^{10} + 2m+3) \cdot 3 + 5 \right) \cdot k
        \\
        &+ \left( (500\cdot 3m \cdot m^{10}+3) \cdot 3+5 \right) \cdot (n-k)
        + (n-1)\\
        =& k\cdot(\sizeMS-1)+(n-k)\cdot(\sizeML-1) +(n-1).
    \end{align*}
    \end{sloppypar}
    \item Place one connected component in each separator gadget $S$, each of size $m^{100}+3$.
\end{itemize}

Next, we describe the positions of the connected components after reconfiguration.

\begin{itemize}
    \item The connected components in the $n$ independent-set gadgets remain in their original positions.
    \item The connected components in the $n$ vertex-cover gadgets remain in their original positions.
    \item Place the connected components that were originally in the $k$ independent-set padding gadgets on the spine of the primary gadget so that no two of them share a vertex or are adjacent. Each such connected component has size $(500\cdot 3m\cdot m^{10}+2m+3)\cdot 3+5 =\sizeMS -1$.
    \item Place the connected components that were originally in the $n-k$ vertex-cover padding gadgets on the spine of the primary gadget so that no two of them share a vertex or are adjacent. Each such connected component has size $(500\cdot 3m\cdot m^{10}+3)\cdot 3+5 = \sizeML - 1$.
    \item Place the connected component that was originally in the primary gadget $L_1$ so that it spans the spines of all independent-set padding gadgets and vertex-cover padding gadgets.
    \item The connected components in the separator gadgets $S$ remain in their original positions.
\end{itemize}

\section{Proof of Correctness}
In this section, we prove the following two lemmas in order to show that the instance $\Phi'$ constructed in the previous section yields a correct reduction.
\begin{lemma}
\label{lma:yes}
For an instance $\Phi=(G,k)$ of the Independent Set problem on a $3$-regular graph, let $\Phi'=(G',A,B)$ be the instance of \CCRCJ~constructed in Section~\ref{sec:whole}. If $G$ has an independent set of size $k$, then $A$ can be reconfigured into $B$ in $G'$.
\end{lemma}
\begin{lemma}
\label{lma:no}
For an instance $\Phi=(G,k)$ of the Independent Set problem on a $3$-regular graph, let $\Phi'=(G',A,B)$ be the instance of \CCRCJ~constructed in Section~\ref{sec:whole}. If $G$ has no independent set of size $k$, then $A$ cannot be reconfigured into $B$ in $G'$.
\end{lemma}

Theorem~\ref{thm:main} follows immediately from the two lemmas above. In what follows, we prove Lemmas~\ref{lma:yes} and~\ref{lma:no}.

\subsection{Proof of Lemma~\ref{lma:yes}}
\begin{proof}[Proof of Lemma~\ref{lma:yes}]
Let $T$ be an independent set of size $k$ in the instance $\Phi$. If there are multiple such independent sets, choose one arbitrarily.
We show that $A$ can be reconfigured into $B$ by exchanging the connected component in the primary gadget with the connected components in the padding gadgets.

\begin{enumerate}
    \item For each vertex $v_i$, proceed as follows.
    \begin{itemize}
        \item If $v_i\in T$, then for each $j\in\{1,2,3\}$, move the connected component placed on $p_{4j-3}$ of $I_i$ together with its attached leaves to the vertex set consisting of $p_1,p_2,p_3$ and their attached leaves in $Ed_{\pi^{-1}(e_{v_i,j}),a}$ when $\rho_{e_{v_i,j}}(v_i)=0$, and to the vertex set consisting of $p_3,p_4,p_5$ and their attached leaves in $Ed_{\pi^{-1}(e_{v_i,j}),a}$ when $\rho_{e_{v_i,j}}(v_i)=1$. Also, move the connected component placed on $p_{4j-1}$ of $I_i$ together with its attached leaves to the vertex set consisting of $p_3,p_4,p_5$ and their attached leaves in $Ed_{\pi^{-1}(e_{v_i,j}),b}$ when $\rho_{e_{v_i,j}}(v_i)=0$, and to the vertex set consisting of $p_1,p_2,p_3$ and their attached leaves in $Ed_{\pi^{-1}(e_{v_i,j}),b}$ when $\rho_{e_{v_i,j}}(v_i)=1$.

        \item If $v_i\notin T$, then for each $j\in\{1,2,3\}$, move the connected component placed on $p_{4j-3}$ of $VC_i$ together with its attached leaves to the vertex set consisting of $p_1$ and its attached leaves in $Ed_{\pi^{-1}(e_{v_i,j}),a}$ when $\rho_{e_{v_i,j}}(v_i)=0$, and to the vertex set consisting of $p_5$ and its attached leaves in $Ed_{\pi^{-1}(e_{v_i,j}),a}$ when $\rho_{e_{v_i,j}}(v_i)=1$. Also, move the connected component placed on $p_{4j-1}$ of $VC_i$ together with its attached leaves to the vertex set consisting of $p_5$ and its attached leaves in $Ed_{\pi^{-1}(e_{v_i,j}),b}$ when $\rho_{e_{v_i,j}}(v_i)=0$, and to the vertex set consisting of $p_1$ and its attached leaves in $Ed_{\pi^{-1}(e_{v_i,j}),b}$ when $\rho_{e_{v_i,j}}(v_i)=1$.
    \end{itemize}
    In every case, the destination vertex set is connected and has exactly the same size as the moved connected component. Hence all these moves are valid.

    \item For each vertex $v_i$, proceed as follows.
    \begin{itemize}
        \item If $v_i\in T$, move one connected component from an independent-set padding gadget so that it occupies all vertices of $I_i$. The moved connected component has size $\sizeMS-1$, and the gadget $I_i$ also contains exactly $\sizeMS-1$ vertices.
        \item If $v_i\notin T$, move one connected component from a vertex-cover padding gadget so that it occupies all vertices of $VC_i$. The moved connected component has size $\sizeML-1$, and the gadget $VC_i$ also contains exactly $\sizeML-1$ vertices.
    \end{itemize}
    Therefore, every move in this step is possible.

    \item Move the connected component in the primary gadget $L_1$ to the padding padding spine region. Its size is $k(\sizeMS-1) + (n-k)(\sizeML-1) + (n-1)$,
    which is exactly the number of spine veritices in the padding vertices in the padding gadgets. Hence this move is possible.

    \item Move the connected components that were initially placed in the padding gadgets into the now-empty primary gadget $L_1$. Their total size is
    $k(\sizeMS-1) + (n-k)(\sizeML-1)$,
    and they require $n-1$ empty vertices between consecutive connected components. Thus, the total number of required vertices is
    $k(\sizeMS-1) + (n-k)(\sizeML-1) + (n-1)$,
    which is exactly the size of $L_1$. Hence this step is also possible.

    \item Return every connected component moved in Step~1 from the edge gadgets to its original position.
\end{enumerate}

After these steps, the resulting configuration is exactly $B$. Therefore, $A$ can be reconfigured into $B$ in $G'$, as claimed.
\end{proof}

\subsection{Proof of Lemma~\ref{lma:no}}
In this section, we prove Lemma~\ref{lma:no}. We use the following auxiliary lemmas.
\begin{lemma}
\label{lma:4}
Under the reconfiguration rule of \CCRCJ, a connected component on a separator gadget cannot be moved to a position that does not contain the spine vertex it contains before reconfiguration.
\end{lemma}
\begin{proof}
We prove this by contradiction. Suppose that such a move is possible; that is, some connected component on a separator gadget can be moved to a position that does not contain the spine vertex it contains before reconfiguration. Let $S_1$ be the first connected component moved in this way.
Since $S_1$ is the first such connected component, every other connected component on a separator gadget still contains its original spine vertex at the moment when $S_1$ is moved. Hence no other separator gadget has enough available vertices to accommodate $S_1$. Therefore, if $S_1$ is moved to a position that does not contain its original spine vertex, then it must be moved into a region outside the separator gadgets.
On the other hand, because there is a separator gadget between every pair of gadgets except for the padding gadgets, no connected component can move so as to span multiple gadgets. Thus, after removing all vertices of the separator gadgets from the caterpillar graph, the size of the largest connected vertex set in the remaining graph is at most $O(m^{11})$. This is overwhelmingly smaller than the size $m^{100}$ of a connected component on a separator gadget. Hence $S_1$ cannot be placed in any connected vertex set outside the separator gadgets.
Therefore, there is no position to which $S_1$ can be moved without containing the spine vertex it contains before reconfiguration. This is a contradiction.
\end{proof}

\begin{definition}[Empty Gadget]
Let $V_X$ be the vertex set of an arbitrary gadget $X$, and let $V_c$ be the vertex set of a connected component $c \in C(U)$. For any $c\in C(U)$, if $V_X \cap V_c = \emptyset$, then we say that the gadget $X$ is empty.
\end{definition}
Let $\gamma$ denote a configuration in which the primary gadget is moved into the padding gadgets. If $A$ can be reconfigured into $B$, then such a configuration must appear.  In configuration $\gamma$, moving the primary gadget there requires the padding gadgets to be empty. The next lemma shows that the empty independent-set gadgets in $\gamma$ encode an independent set.

\begin{lemma}
\label{lma:5}
In the configuration $\gamma$, the total number of empty vertices contained in the independent-set gadgets, the vertex-cover gadgets, and the edge gadgets is at most $O(m^{2})$.
\end{lemma}
\begin{proof}
Let $A$ be the total number of vertices contained in the independent-set gadgets, the vertex-cover gadgets, and the edge gadgets.

By the construction of the gadgets, $A$ is the sum of the following quantities:
\begin{enumerate}
    \item the total number of vertices in the independent-set gadgets:
    \[
    \sum |I_i| = \left((500\cdot 3m\cdot m^{10}+2m+3)\cdot 3+5\right)\cdot n,
    \]
    \item the total number of vertices in the vertex-cover gadgets:
    \[
    \sum |VC_i| = \left((500\cdot 3m\cdot m^{10}+3)\cdot 3+5\right)\cdot n,
    \]
    \item the total number of vertices in the edge gadgets:
    \[
    \sum |Ed_i| = \left(500\cdot 3m\cdot m^{10}+m+4\right)\cdot 2m.
    \]
\end{enumerate}

Next, let $B$ denote the total size of the connected components placed in the independent-set gadgets, the vertex-cover gadgets, and the edge gadgets in the configuration $\gamma$.
In the configuration $\gamma$, the connected component originally in the primary gadget is placed in the padding gadgets, and the separator components remain in the separators. Hence, the connected components that occupy the independent-set gadgets, the vertex-cover gadgets, and the edge gadgets in the configuration $\gamma$ are precisely the following:

\begin{enumerate}
    \item the connected components originally placed in the independent-set gadgets, which are still placed in the independent-set gadgets or the edge gadgets in the configuration $\gamma$, and whose total size is
    \[
    \sum |I_i|-5n = \left((500\cdot 3m\cdot m^{10}+2m+3)\cdot 3\right)\cdot n,
    \]
    \item the connected components originally placed in the vertex-cover gadgets, which are still placed in the vertex-cover gadgets or the edge gadgets in the configuration $\gamma$, and whose total size is
    \[
    \sum |VC_i|-5n = \left((500\cdot 3m\cdot m^{10}+3)\cdot 3\right)\cdot n,
    \]
    \item the connected components originally placed in the padding gadgets, which are placed in the independent-set gadgets or the vertex-cover gadgets in the configuration $\gamma$, and whose total size is
    \begin{align*}
    \sum |MS+ML|-(n-1)
    &=& \left((500\cdot 3m\cdot m^{10}+2m+3)\cdot 3+5\right)\cdot k\\
    & &+ \left((500\cdot 3m\cdot m^{10}+3)\cdot 3+5\right)\cdot (n-k).
    \end{align*}
\end{enumerate}

Therefore, the number of empty vertices contained in the independent-set gadgets, the vertex-cover gadgets, and the edge gadgets is at most $A-B$.

Since the graph is $3$-regular, we have $2m=3n$. Hence,
\begin{align}
A-B
&= \Bigl( \left((500 \cdot 3m \cdot m^{10} + 2m + 3) \cdot 3 + 5\right) \cdot n \notag\\
&\quad + \left((500 \cdot 3m \cdot m^{10} + 3) \cdot 3 + 5\right) \cdot n \notag\\
&\quad + \left(500 \cdot 3m \cdot m^{10} + m + 4\right) \cdot 2m \Bigr) \notag\\
&\quad - \Bigl( \left((500 \cdot 3m \cdot m^{10} + 2m + 3) \cdot 3\right) \cdot n \notag\\
&\quad + \left((500 \cdot 3m \cdot m^{10} + 3) \cdot 3\right) \cdot n \notag\\
&\quad + \left((500 \cdot 3m \cdot m^{10} + 2m + 3) \cdot 3 + 5\right) \cdot k \notag\\
&\quad + \left((500 \cdot 3m \cdot m^{10} + 3) \cdot 3 + 5\right) \cdot (n-k) \Bigr) \notag\\
&= \left(6m+5\right)n + 5n + \left(500 \cdot 3m \cdot m^{10} + m + 4\right) \cdot 2m \notag\\
&\quad - \left((500 \cdot 3m \cdot m^{10} + 2m + 3) \cdot 3 + 5\right) \cdot k \notag\\
&\quad - \left((500 \cdot 3m \cdot m^{10} + 3) \cdot 3 + 5\right) \cdot (n-k) \notag\\
&= (6m+10)n + \left(500 \cdot 3m \cdot m^{10} + m + 4\right) \cdot 2m \notag\\
&\quad - \left((500 \cdot 3m \cdot m^{10} + 3) \cdot 3 + 5\right)n - 6mk \notag\\
&= (6m+5)n + \left(500 \cdot 3m \cdot m^{10} + m + 4\right) \cdot 2m \notag\\
&\quad - \left(500 \cdot 3m \cdot m^{10} + 3\right) \cdot 3n - 6mk \notag\\
&= (6m+5)n + 2m^2 + 8m - 9n - 6mk \notag\\
&= 6mn + 2m^2 - 4n + 8m - 6mk \notag\\
&= 2m^2 + \frac{16}{3}m - 6mk,
\end{align}
Since $k \le n = O(m)$, we obtain $A-B = O(m^2)$.
\end{proof}

\begin{lemma}
    \label{lma:6}
For each $j\in\{1,2,3\}$, let $Ed_j^i$ denote the edge gadget corresponding to $e_{v_i,j}$.
Also, within $Ed_j^i$, let $Ed_{j,a}^i$ denote the caterpillar consisting of $P_a$ and its attached leaves, and let $Ed_{j,b}^i$ denote the caterpillar consisting of $P_b$ and its attached leaves.
Then, for the independent-set gadget $I_i$, the following holds:
for any $j\in\{1,2,3\}$, the connected component that is initially placed on $p_{4j-3}$ and its attached leaves is, in the configuration $\gamma$, placed either so as to contain $p_{4j-3}$ or on $Ed_{j,a}^i$.
\end{lemma}
\begin{proof}
Fix an arbitrary $j\in\{1,2,3\}$, and let $C$ be the connected component that is initially placed on the vertex $p_{4j-3}$ of $I_i$ together with its attached leaves.
We prove the claim by contradiction.
Suppose that, in the configuration $\gamma$, $C$ is placed so as not to contain $p_{4j-3}$ and is not placed on $Ed_{j,a}^i$.

In the configuration $\gamma$, every connected component except the one originally placed in the primary gadget is contained in an independent-set gadget, a vertex-cover gadget, or an edge gadget.
Hence, $C$ must be placed in one of the following:
a vertex-cover gadget, an independent-set gadget not containing $p_{4j-3,i}$, or an edge gadget other than $Ed_{j,a}^i$.

Now, among the spine vertices of the independent-set gadgets and the vertex-cover gadgets, the only vertices with at least $m^{10}$ attached leaves are $p_{4t-3}$ and $p_{4t-1}$ for $t\in\{1,2,3\}$.
Similarly, among the spine vertices of the edge gadgets, the only such vertices are $p_{1,a},p_{5,a},p_{1,b},p_{5,b}$.
Moreover, by Lemma~\ref{lma:5}, the total number of empty vertices in the configuration $\gamma$ is at most $O(m^2)$.
Therefore, if $C$ were placed so as to contain two or more such vertices, then at least $\Omega(m^{10})$ empty vertices would be created, contradicting Lemma~\ref{lma:5}.
Hence, $C$ cannot be placed so as to contain two or more such vertices.

\noindent
\textbf{Case 1: $C$ is placed in a vertex-cover gadget.}

Let $VC_r$ be the vertex-cover gadget in which $C$ is placed.
By the argument above, $C$ must be placed so as to contain exactly one of $p_{4t-3}$ and $p_{4t-1}$ in $VC_r$, for some $t\in\{1,2,3\}$.

The size of $C$ is
\[
500\pi^{-1}(e_{v_i,j})m^{10}+m+1+\rho_{e_{v_i,j}}(v_i).
\]
On the other hand, the size of a candidate destination vertex set in $VC_r$ is either
\[
500\pi^{-1}(e_{v_r,t})m^{10}+1+\rho_{e_{v_r,t}}(v_r)
\]
or
\[
500(3m-\pi^{-1}(e_{v_r,t}))m^{10}+2-\rho_{e_{v_r,t}}(v_r).
\]

If $(i,j)\neq(r,t)$, then the absolute difference between these sizes is at least $500m^{10}-m-1$.
Thus, if $C$ is smaller, then $\Omega(m^{10})$ empty vertices are created; if $C$ is larger, then $C$ must contain two or more vertices having at least $m^{10}$ attached leaves.
Both are impossible.

If $(i,j)=(r,t)$, then the size of $C$ is still larger than that of the corresponding vertex set on the $p_{4j-3}$ side by $m$, and is even larger than that of the corresponding vertex set on the $p_{4j-1}$ side.
Hence $C$ cannot be placed there, either.

Therefore, Case 1 cannot occur.

\noindent
\textbf{Case 2: $C$ is placed in an independent-set gadget not containing $p_{4j-3,i}$.}

Let $I_r$ be the independent-set gadget in which $C$ is placed.
Again, $C$ must be placed so as to contain exactly one of $p_{4t-3}$ and $p_{4t-1}$ in $I_r$, for some $t\in\{1,2,3\}$.

The size of $C$ is
\[
500\pi^{-1}(e_{v_i,j})m^{10}+m+1+\rho_{e_{v_i,j}}(v_i),
\]
whereas the size of a candidate destination vertex set in $I_r$ is either
\[
500\pi^{-1}(e_{v_r,t})m^{10}+m+1+\rho_{e_{v_r,t}}(v_r)
\]
or
\[
500(3m-\pi^{-1}(e_{v_r,t}))m^{10}+m+2-\rho_{e_{v_r,t}}(v_r).
\]

If $(i,j)\neq(r,t)$, then the absolute difference between these sizes is at least $500m^{10}-1$.
Hence, as in Case 1, either $\Omega(m^{10})$ empty vertices are created or $C$ must contain two or more vertices having at least $m^{10}$ attached leaves.
Both are impossible.

Next, consider the case $(i,j)=(r,t)$.
Since $C$ does not contain $p_{4j-3,i}$ by assumption, the only remaining possibility inside $I_i$ is that $C$ is placed on the $p_{4j-1,i}$ side.
However, the size of the vertex set consisting of $p_{4j-1,i}$ and its attached leaves is
\[
500(3m-\pi^{-1}(e_{v_i,j}))m^{10}+m+2-\rho_{e_{v_i,j}}(v_i),
\]
and thus the difference from the size of $C$ is
\[
500(3m-2\pi^{-1}(e_{v_i,j}))m^{10}+1-2\rho_{e_{v_i,j}}(v_i).
\]
This is $\Omega(m^{10})$, contradicting Lemma~\ref{lma:5}.

Therefore, Case 2 cannot occur.

\noindent
\textbf{Case 3: $C$ is placed in an edge gadget other than $Ed_{j,a}^i$.}

First, suppose that $C$ is placed in $Ed_{j,b}^i$.
Then $C$ must be placed so as to contain exactly one of $p_{1,b}$ and $p_{5,b}$.
The sizes of the corresponding vertex sets are
\[
500(3m-\pi^{-1}(e_{v_i,j}))m^{10}+1
\quad \text{or} \quad
500(3m-\pi^{-1}(e_{v_i,j}))m^{10}+2.
\]
Hence the size difference from $C$ is at least
\[
500(3m-2\pi^{-1}(e_{v_i,j}))m^{10}-m,
\]
which is $\Omega(m^{10})$.
This contradicts Lemma~\ref{lma:5}.

Next, suppose that $C$ is placed in $Ed_r$ for some $r\neq \pi^{-1}(e_{v_i,j})$.
Again, $C$ must be placed so as to contain exactly one of
$p_{1,a},p_{5,a},p_{1,b},p_{5,b}$.
The size of a candidate destination vertex set is one of
\[
500rm^{10}+1,\quad 500rm^{10}+2,\quad
500(3m-r)m^{10}+1,\quad 500(3m-r)m^{10}+2.
\]
Since $r\neq \pi^{-1}(e_{v_i,j})$, the absolute difference between these sizes and the size of $C$ is at least $500m^{10}-m-1$.
Thus, as in Case 1, either $\Omega(m^{10})$ empty vertices are created or $C$ must contain two or more vertices having at least $m^{10}$ attached leaves.
Both are impossible.

Therefore, Case 3 cannot occur.

Since all cases lead to contradictions, the assumption is false.
Therefore, in the configuration $\gamma$, the connected component that is initially placed on $p_{4j-3}$ and its attached leaves is placed either so as to contain $p_{4j-3}$ or on $Ed_{j,a}^i$.
\end{proof}

We are now ready to prove Lemma~\ref{lma:7}. Lemma~\ref{lma:5} bounds the number of empty vertices in the configuration $\gamma$, and Lemma~\ref{lma:6} uses this bound to restrict where a connected component from an independent-set gadget can be placed.

\begin{lemma}
\label{lma:7}
In the configuration $\gamma$, the two independent-set gadgets $I_i$ and $I_j$ corresponding to the endpoints of an edge $e_k=(v_i,v_j)$ cannot both be empty.
\end{lemma}
\begin{proof}
We prove this by contradiction.
Suppose that, in the configuration $\gamma$, both $I_i$ and $I_j$ are empty.
Let $h,l\in\{1,2,3\}$ be such that $e_k=e_{v_i,h}=e_{v_j,l}$.

Since $I_i$ and $I_j$ are empty, the connected components that were initially placed on
$p_{4h-3,i}$ and its attached leaves, and on $p_{4l-3,j}$ and its attached leaves, do not contain their original vertices in $\gamma$.
Therefore, by Lemma~\ref{lma:6}, both of these connected components must be placed on $Ed_{k,a}$.

The number of vertices in $Ed_{k,a}$ is $1000km^{10}+m+4$.
On the other hand, the sum of the sizes of the above two connected components is
\[
(500km^{10}+m+1+\rho_{e_k}(v_i))
+(500km^{10}+m+1+\rho_{e_k}(v_j))
=1000km^{10}+2m+3,
\]
because $\rho_{e_k}(v_i)+\rho_{e_k}(v_j)=1$.

Since $1000km^{10}+2m+3 > 1000km^{10}+m+4$,
these two connected components cannot be placed on $Ed_{k,a}$ without overlapping.
This is a contradiction.
\end{proof}

Using the lemmas above, we now prove Lemma~\ref{lma:no}.

\begin{proof}[Proof of Lemma~\ref{lma:no}]    
We prove the contrapositive; that is, we show that if $A$ can be reconfigured into $B$ in $G'$, then $G$ has an independent set of size at least $k$.
Assume that $A$ can be reconfigured into $B$ in $G'$.
By Lemma~\ref{lma:4}, every connected component initially placed in a separator must continue to contain its original separator vertex throughout the reconfiguration. If the reconfiguration from $A$ to $B$ is possible, then at some point there must appear a configuration $\gamma$ in which the connected component in the primary gadget is moved into the padding gadgets.
Since the primary component is placed in the padding gadgets in the configuration $\gamma$, all padding gadgets are empty.
Let $S'=\{\,v_i \mid I_i \text{ is occupied by a connected component from } MS \text{ in } \gamma\,\}$.
We show that $S'$ is an independent set of size at least $k$ in $G$. 
The $k$ connected components initially placed in the independent-set padding gadgets must be moved out of the padding gadgets.
Each such connected component has size $\sizeMS-1$, which is exactly the number of vertices of an independent-set gadget.
Therefore, each of them occupies one independent-set gadget entirely, and so $|S'|=k$.

It remains to show that $S'$ is independent.
Suppose, to the contrary, that there exist adjacent vertices $v_i,v_j\in S'$ with $e_k=(v_i,v_j)\in E(G)$.
Then both $I_i$ and $I_j$ are occupied by connected components initially placed in independent-set padding gadgets, and hence both gadgets are empty with respect to the components initially placed there.
This contradicts Lemma~\ref{lma:7}. Hence $S'$ is an independent set of size at least $k$. This proves the contrapositive of Lemma~\ref{lma:no}, and the lemma follows.
\end{proof}
\begin{proof}[Proof of Theorem~\ref{thm:main}]
By Lemmas~\ref{lma:yes} and~\ref{lma:no}, given an instance $\Phi=(G,k)$ of the Independent Set problem on a $3$-regular graph, we can construct an instance of \CCRCJ~on a caterpillar graph such that $G$ has an independent set of size $k$ if and only if the reconfiguration is possible. Therefore, Theorem~\ref{thm:main} follows.
\end{proof}

\section{\CCRCJ~for Path Graphs}

In this section, we study the decision problem for \CCRCJ~on path graphs and prove Theorem~\ref{thm:intro-path}. We also present a constructive result for instances with sufficiently large empty space.

Let $(G,A,B)$ be an instance of \CCRCJ, and suppose that $G$ is a path graph. Let $v_1,\dots,v_n$ be the vertices of the path ordered from left to right. For any vertex set $U\subseteq V$, let $h_U$ denote the sequence of the sizes of the connected components of $U$, listed from left to right. To distinguish connected components of the same size, we introduce the sequence $h'_U$, obtained by appending indices according to their order of appearance. Let $h'_A=(x_1,\dots,x_k)$ and $h'_B=(y_1,\dots,y_k)$. For an element $z$ in $h'_A$, let $\sigma_A(z)$ denote the position of $z$ in $h'_A$, and similarly let $\sigma_B(z)$ denote the position of $z$ in $h'_B$. Define
\[
\operatorname{inv}(h'_A,h'_B)
=\{(x_i,y_j)\mid \sigma_A(x_i)<\sigma_A(y_j),\ \sigma_B(x_i)>\sigma_B(y_j)\}.
\]
That is, $\operatorname{inv}(h'_A,h'_B)$ is the set of pairs of connected components whose relative order must be reversed when transforming $A$ into $B$. Let $k=|C(A)|=|C(B)|$, and define $b(A,k)=n-|A|-k$.

Nakahata~\cite{nakahata2025reconfiguring} proved the following characterization and an $O(n^2)$-time algorithm for the decision problem on path graphs.

\begin{lemma}[\cite{nakahata2025reconfiguring}]
\label{lma:path}
$A$ and $B$ are reconfigurable under \CCRCJ~if and only if
\begin{align}
\max_{(x_i,y_j)\in \operatorname{inv}(h'_A,h'_B)} \min\{x_i,y_j\} \leq b(A,k)
\end{align}
holds.
\end{lemma}

We improve the above $O(n^2)$-time algorithm by proving Theorem~\ref{thm:intro-path}. We also prove a constructive result for instances with sufficiently large empty space.
\begin{theorem}
\label{thm:myconstruct}
    Suppose that $G$ is a path graph.
    If the answer to the \CCRCJ~instance $(G,A,B)$ is YES and the following inequality holds, then there exists a reconfiguration sequence from $A$ to $B$ of length $O(n\log n)$, and such a sequence can be output in $O(n\log n)$ time:
    \[
    \max_{(x_i,y_j)\in \operatorname{inv}(h'_A,h'_B)} 2\max\{x_i,y_j\}\le b(A,k).
    \]
\end{theorem}

We first prove Theorem~\ref{thm:intro-path}.
\begin{proof}[Proof of Theorem~\ref{thm:intro-path}]
    By Lemma~\ref{lma:path}, it suffices to determine whether inequality (1) holds in order to decide whether \CCRCJ~is a yes-instance.
    The sequences $h'_A$, $h'_B$, and the value $b(A,k)$ can all be computed in $O(n)$ time.
    Therefore, it remains to compute
    \[
    \max_{(x_i,y_j)\in \operatorname{inv}(h'_A,h'_B)} \min\{x_i,y_j\}
    \]
    in $O(n\log n)$ time.

    For any $1\le i\le k$, we say that an element $x_i$ \emph{contributes to the decision} if there exists an element $y$ such that $(x_i,y)\in \operatorname{inv}(h'_A,h'_B)$ and $x_i\le y$.
    Similarly, for any $1\le j\le k$, we say that an element $y_j$ \emph{contributes to the decision} if there exists an element $x$ such that $(x,y_j)\in \operatorname{inv}(h'_A,h'_B)$ and $y_j\le x$.
    Then
    \[
    \max_{(x_i,y_j)\in \operatorname{inv}(h'_A,h'_B)} \min\{x_i,y_j\}
    \]
    is exactly the maximum among all elements $x_i$ and $y_j$ that contribute to the decision.

    We explain a method based on a segment tree for determining, for all $1\le i\le k$, whether $x_i$ contributes to the decision.
    First, prepare an array $M$ of length $k$ indexed by the positions in $h'_A$.
    Let $M[a]$ denote the $a$-th entry of $M$.
    Initially, all entries are set to $-\infty$.
    Process the elements of $h'_B$ in increasing order of their positions in $h'_B$.
    Let the currently processed element be $x$.
    For each element $x$, perform the following operations:
    \begin{enumerate}
        \item Compute the maximum value $y^*$ in the interval $[\sigma_A(x)+1,k]$ of $M$.
        If the interval $[\sigma_A(x)+1,k]$ is empty, return $-\infty$.
        At this moment, only elements $y$ satisfying $\sigma_B(y)<\sigma_B(x)$ have been reflected in $M$.
        If $x\le y^*$, then we determine that $x$ contributes to the decision; otherwise, it does not.
        By using a range maximum query on a segment tree, this computation can be done in $O(\log k)$ time.
        \item Update $M[\sigma_A(x)]$ to $x$.
    \end{enumerate}
    Thus, whether each $x_i$ contributes to the decision can be determined in $O(n\log n)$ time.
    The same argument applies to the elements $y_j$.
    Therefore, the \CCRCJ~decision problem can be solved in $O(n\log n)$ time.
\end{proof}

Next, we prove Theorem~\ref{thm:myconstruct}.
\begin{proof}[Proof of Theorem~\ref{thm:myconstruct}]
Let $h'_A=(x_1,\dots,x_k)$ be the left-to-right order of the connected components in the initial configuration $A$, and let $h'_B=(y_1,\dots,y_k)$ be the left-to-right order of the connected components in the target configuration $B$.

We describe a recursive procedure that transforms the order $h'_A$ into $h'_B$.
More generally, for any contiguous subsequence $S$ of the current configuration, let $B[S]$ denote the order induced by $h'_B$ on the connected components contained in $S$.
Our recursive procedure rearranges $S$ into the order $B[S]$.

Consider such a subsequence $S$ containing $m$ connected components.
Let $B[S]=(z_1,\dots,z_m)$, and partition these components into
\[
S_{\mathrm{left}}=\{z_i \mid 1\le i\le \lfloor m/2\rfloor\},
\qquad
S_{\mathrm{right}}=\{z_i \mid \lfloor m/2\rfloor < i\le m\}.
\]
Thus, $S_{\mathrm{left}}$ is the set of components that should appear in the first half of the target order $B[S]$, and $S_{\mathrm{right}}$ is the set of components that should appear in the second half.

We first show how to rearrange $S$ so that all components in $S_{\mathrm{left}}$ appear before all components in $S_{\mathrm{right}}$. Let
\[
L(S):=\max_{(x_i,y_j)\in \operatorname{inv}(h'_A,h'_B),\ x_i,y_j\in S} \max\{x_i,y_j\}.
\]
Since every inversion pair formed by components in $S$ is also an inversion pair of the whole instance, we have
\[
L(S)\le \max_{(x_i,y_j)\in \operatorname{inv}(h'_A,h'_B)} \max\{x_i,y_j\}.
\]
Therefore, by the assumption of the theorem, $2L(S)\le b(A,k)$.
Hence we can reserve empty space of total length at least $2L(S)$ on the two sides of the current subsequence.

We maintain three contiguous intervals:
the middle interval, consisting of components not yet moved;
the left interval, consisting of components already moved and belonging to $S_{\mathrm{left}}$;
and the right interval, consisting of components already moved and belonging to $S_{\mathrm{right}}$.
Initially, the whole subsequence $S$ is the middle interval, and the left and right intervals are empty.

At every step, we take one component from either end of the middle interval.
If the leftmost component belongs to $S_{\mathrm{left}}$, we move it to the left interval.
If the rightmost component belongs to $S_{\mathrm{right}}$, we move it to the right interval.
Thus, the only remaining case is that the leftmost component $x$ belongs to $S_{\mathrm{right}}$ and the rightmost component $y$ belongs to $S_{\mathrm{left}}$.

In this case, since $x$ should appear in the right half of $B[S]$ but is currently at the left end of the middle interval, there exists a component $x'\in S_{\mathrm{left}}$ still contained in the middle interval such that $(x,x')$ is an inversion pair.
Similarly, there exists a component $y'\in S_{\mathrm{right}}$ such that $(y',y)$ is an inversion pair.
Therefore,
\[
x\le L(S),\qquad y\le L(S).
\]
Since the total empty space on the two sides has length at least $2L(S)$, at least one side has empty space of length at least $L(S)$.
If the left side has empty space of length at least $L(S)$, then $y$ can be moved to the left interval.
If the right side has empty space of length at least $L(S)$, then $x$ can be moved to the right interval.
Each of these moves takes a component at an endpoint of the middle interval and places it into empty space on one side, without crossing any other component, and hence it can be realized by one valid move under \CCRCJ.
Thus, in either case, one valid move is possible.

Therefore, one component can be removed from the middle interval in each step, and each component is moved at most once in this partition step.
Hence the partition step requires $O(m)$ moves.
After it finishes, the left interval contains exactly the components of $S_{\mathrm{left}}$, and the right interval contains exactly the components of $S_{\mathrm{right}}$.

Next, we recursively rearrange the left interval into the order induced by $B[S]$ on $S_{\mathrm{left}}$, and then recursively rearrange the right interval into the order induced by $B[S]$ on $S_{\mathrm{right}}$.
These recursive calls are performed sequentially, not simultaneously.
Hence, at any moment, workspace is needed for only one subproblem.
Before each recursive call, we compress the components on the other side together with the remaining empty vertices so as to secure empty space of total length at least $2L(S)$ around the current subproblem.
This rearrangement can be done in $O(m)$ moves.

Let $T(m)$ denote the maximum number of moves required to rearrange any subsequence of length $m$ into its target order.
From the discussion above, we obtain
\[
T(m)=T(\lfloor m/2\rfloor)+T(\lceil m/2\rceil)+O(m).
\]
Therefore,
\[
T(m)=O(m\log m).
\]
Since $m\le k\le n$, the total number of moves is $O(n\log n)$.

Finally, the sequence itself can be output within the same time bound.
At each recursive stage, we only need to determine the membership of each component in $S_{\mathrm{left}}$ or $S_{\mathrm{right}}$, execute the partition step, and recurse on the two resulting subsequences.
Each stage takes linear time in the size of the current subsequence, and therefore the whole reconfiguration sequence can be output in $O(n\log n)$ time.
\end{proof}

\bibliography{reference}

\end{document}